\documentclass[journal, usletter]{IEEEtran}

\usepackage{amsmath,amsthm}
\usepackage{graphics, subfigure}
\usepackage{tikz}
\usepackage{mathtools}
\usetikzlibrary{shapes,arrows}
\usepackage{algorithm, algorithmic}
\usepackage{soul}
\usepackage[outdir=./]{epstopdf}
\usepackage{color}
\usepackage{multirow}
\usetikzlibrary{automata,positioning}
\usepackage{wrapfig}
\usepackage{rotating}
\usepackage{bbm}
\usepackage{url}

\usepackage{enumitem}

\newtheorem{assumption}{Assumption}
\newtheorem{theorem}{Theorem}

\newtheorem{proposition}{Proposition}
\newtheorem{definition}{Definition}
\newtheorem{remark}{Remark}
\newtheorem{corollary}{Corollary}

\IEEEoverridecommandlockouts
\title{A Robust Algorithm for Tile-based 360-degree Video Streaming with Uncertain FoV Estimation}

\author{Arnob Ghosh, Vaneet Aggarwal,\thanks{The authors are with the School of IE, Purdue University, West Lafayette, IN 47907, email: \{ghosh39,vaneet\}@purdue.edu} and Feng Qian\thanks{Feng is with the Computer Science Department of Indiana University, Bloomington, IN, email: fengqian@indiana.edu}}

\begin{document}
\maketitle
\begin{abstract}
We propose a robust scheme for streaming 360-degree immersive videos to maximize the quality of experience (QoE).
Our streaming approach introduces a holistic analytical framework built upon the formal method
of stochastic optimization. We propose a robust algorithm which provides a streaming rate such that the video quality degrades below that rate with very low probability even in presence of uncertain head movement, and bandwidth. 
It assumes the knowledge of the viewing probability of different portions (tiles) of a panoramic scene. Such probabilities can be easily derived from crowd-sourced
measurements performed by 360 video content providers.
We then propose efficient methods to solve the problem at runtime
while achieving a bounded optimality gap (in terms of the QoE). 
We implemented our proposed approaches using emulation. Using real users' head movement traces and real cellular bandwidth
traces, we show that our algorithms significantly outperform the baseline
algorithms by at least in 30\% in the QoE metric. Our algorithm gives a streaming rate which is 50\% higher compared to the baseline algorithms when the prediction error is high. 
\end{abstract}


\section{Introduction}
360-degree videos provide users a panoramic view and create a unique viewing experience. These videos are recorded using the omnidirectional cameras. While watching the video, the user can change the viewing direction by moving her head. Typically, the user wearing a VR headset (e.g., the Google Cardboard \cite{cardboard}) can adjust her orientation by changing the pitch, yaw, and roll of the device which corresponds to the X, Y and Z axes, respectively (Fig.~\ref{fig:360}). The field-of-view (FoV) defines the extent of the user's observable portion. It is typically fixed for a VR headset (e.g., 90-degree vertically and 110-degree horizontally \cite{feng}).

The mainstream technologies for streaming 360 videos fetch all panoramic scenes including both the visible and invisible portions\cite{fb,youtube}.
Though this is simple, it has some disadvantages. In particular, the bandwidth utilization is high as the chunks of 360-degree videos are of much larger sizes compared to the traditional ones. Thus, if the  bandwidth is low, it will lead to poor viewing quality or stall (rebuffering). Even the wireline capacity may not be enough for HD 360-degree videos \cite{wireline}. Although significant progress has been made in developing VR technologies, the research community still lacks robust bandwidth-efficient streaming algorithms for 360-degree videos for maximizing the QoE of the users.
%
%


\begin{figure}[hbtp]
\vspace{-0.1in}
\begin{center}
	\includegraphics[width=2in]{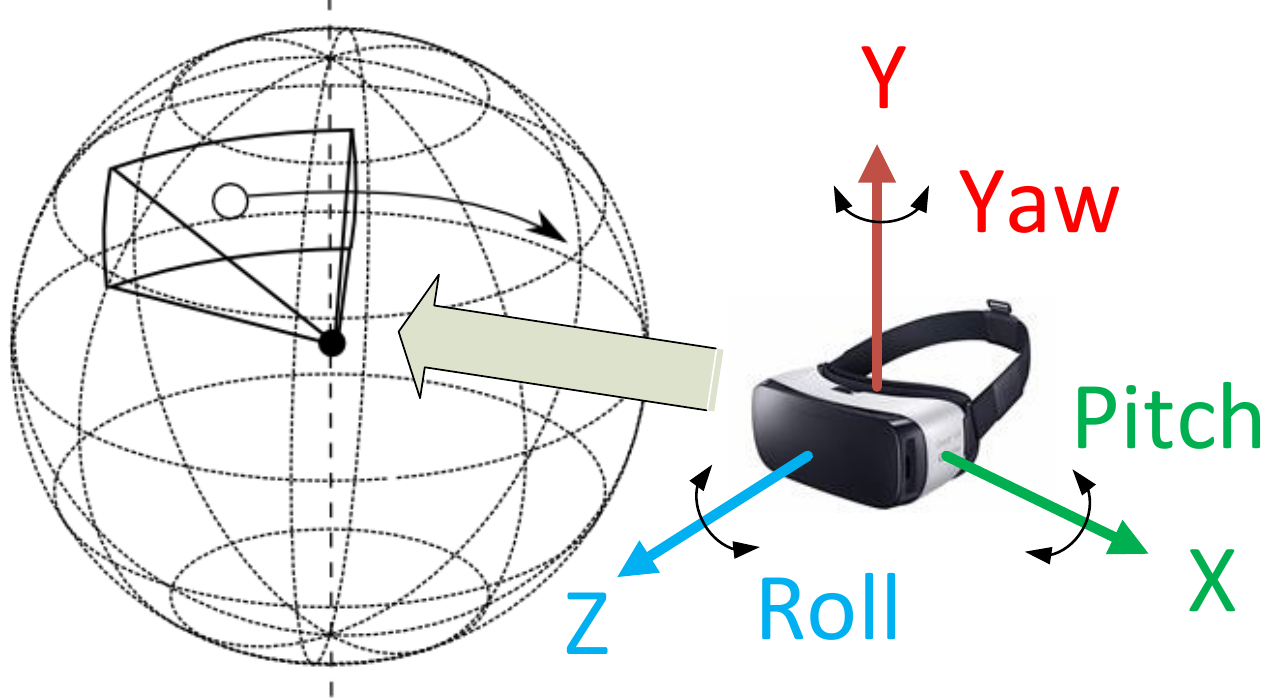}
\end{center}
\vspace{-0.1in}
\caption{Adjust viewing direction during 360 degree video playback}
\label{fig:360}
\vspace{-0.1in}
\end{figure}



In a tile-based video streaming,
instead of downloading the entire panoramic view at a same rate, a video player can download the portions of the chunk which are more likely to be viewed at a higher rate.
This requires spatially segmenting a 360-degree video chunk into multiple segments called \emph{tiles}. A tile (as opposed to a chunk) is the smallest downloadable content unit.
Tile-based 360 video streaming has been investigated by prior studies~\cite{bao,feng,tile1}.  Designing a tile-based streaming algorithm for ABR (adaptive bitrate) 360 videos is very challenging because the algorithm needs to make judicious decisions in both the quality dimension and the spatial (tile) dimension. We note that most existing studies suffer from several limitations.
(1) They mostly take heuristics-based approaches that may lead to suboptimal viewing experiences.
(2) They typically predict \emph{individual} users' head movements to determine which tiles to fetch. Meanwhile, another opportunity for prefetching tiles is to leverage \emph{crowd-sourcing}, i.e., if many other users are interested in viewing a certain direction at certain time, then the corresponding tiles can be prefetched. Despite the simple idea, it is unclear how to incorporate crowd-sourcing into ABR streaming algorithms in a principled manner.
(3) There lacks theoretical rigorousness in terms of maximizing the QoE in the face of fluctuating bandwidth and users' fast head movement.
(4) Bandwidth, and the head movement (thus, FoV) varies randomly over time. There lacks a robust algorithm which can provide a rate such that the streaming rate degrades below that value with very low probability. 


To address the above challenges, our streaming approach introduces a holistic analytical framework built upon the formal method of stochastic optimization. We assume that the user QoE gets a utility based on minimum rate of consumed tiles in each FoV, because the user would like to view all the tiles in the FoV in ideally the same quality to avoid any quality fluctuation. Also the QoE gets penalized when a stall occurs or the inter-chunk quality changes.
We then formulate a stochastic optimization problem by assuming the knowledge of the viewing probability of the tiles. Such probabilities can be easily derived from the combination of crowd-sourced measurements and the FoV prediction of current user performed by 360 video content providers. The formulation jointly
maximizes the consumed bitrate, minimizes the play-back rate variation across chunks, and minimizes the total stall time, based on detailed modeling of the above metrics in the 360 video contexts.

The above optimization problem turns out to be non-convex because of the discrete variable space, and is therefore difficult to solve.
To address this technical challenge, we first relax the constraints and formulate a relaxed problem that is convex.
We then design an algorithm that computes the final scheduling decisions
by strategically mapping the solutions of the relaxed problem back to the original problem. The QoE gap between the solution to the original formulation (the optimality) and our obtained solution is proved to be bounded. Based on the above algorithms, we then design its online version that can adapt to the fluctuating bandwidth and uncertain head movement at runtime with low overhead incurred.
We implemented our proposed algorithm using an emulation-based approach. The emulated video player performs real network transfers of synthetic tiles over TCP/IP. The bandwidth is also emulated using realistic cellular network traces. We use real users' head movement traces to evaluate our proposed algorithms.
Using the traces of 40 users, we compute the empirical distribution of the FoV. We then apply our algorithm to the traces of 10 users.  The emulation results show that our proposed algorithms significantly outperform the baseline algorithms by at least in 30\% in the QoE metric. The baseline algorithms are considered to be the one which fetches all the tiles rather prioritizing those tiles which have higher chances to be viewed. We show that the average rate with which users will watch the video is almost 50\% higher than the baseline algorithms. Our empirical result also shows that our proposed algorithm outperforms the baseline algorithms by at least 60\%  in presence of higher bandwidth uncertainty. 

To summarize, our main contributions are as follows.

\begin{itemize}[leftmargin=*]

\item To the best of our knowledge, we propose the first principled algorithm that employs the formal method of stochastic optimization for streaming 360-degree videos. 
    It can directly work with crowd-sourced viewing statistics to make judicious decisions of prefetching tiles.
    Our approach explicitly takes into account the probabilistic nature of viewers' head movement randomness when watching 360-degree videos. 


\item The above QoE maximization problem is difficult to solve due to its non-convex nature. We propose efficient methods to solve it while achieving a bounded optimality gap (in terms of the QoE) that is independent of the number of chunks. We also design the online versions of the algorithm that dynamically adjust to the fluctuating bandwidth and uncertain head movement at runtime.


\item We implemented our proposed approaches using emulation. Using real users' head movement traces and real cellular bandwidth traces, we show using our emulation testbed that our proposed algorithms can achieve significantly higher QoE , higher average streaming rate compared to the baseline algorithms.
\end{itemize}


\section{Related Literature}Recently, novel encoding schemes for 360-degree video streaming have been proposed \cite{Hosseini1,D'Acunto,occulus}. However, the above papers did not study the optimality of video streaming algorithms in terms of maximizing the perceived QoE  over the limited bandwidth.

The authors of  \cite{yin,jiang} proposed algorithms for online video streaming under limited bandwidth which maximizes the QoE. However, these papers did not consider the {\em 360-degree videos}. The 360-degree videos pose unique challenges and thus require new metric for the QoE. For example,  in the 360-degree video, each chunk consists of several tiles. Hence, a video streaming algorithm now needs to find the rate at which each tile of a chunk has to be downloaded. In contrast, the above papers only need to find the rate at which a chunk has to be downloaded. In the 360-degree video, the FoV depends on the tiles of a chunk a user is viewing, the user may not view all the tiles, or in a viewing some tiles may be of different qualities which can impact the QoE. Hence, a new QoE metric is required for the 360-degree videos depending on the FoV.

 \cite{Xiao:2017,Xie:2017,Corbillon:2017} recently formulated a mathematical model of QoE and proposed heuristic based algorithms. Heuristic based algorithms for 360-degree video streaming have been proposed   \cite{hosseini, ganainy, Graf, Petrangeli}.   \cite{LeFeuvre, feng, bao} proposed an algorithm for streaming of tile-based 360-degree video streaming.  In all the above papers  authors considered the short term prediction of the head movement and fetched tiles only along the predicted FoV or adjacent to the predicted FoV. Thus, if the FoV prediction is wrong, the rate will be very low. Hence, the above papers did not provide any probabilistic guarantee on the rate achieved. Further, the FoV prediction is only accurate for shorter time scale. Hence, these papers did not consider the temporal rate variation across the chunks as the QoE metric. Thus, the rate may be different across different chunks. Apart from the current head position of the users, we also use the FoV from the crowd sourced data to estimate the FoV in a longer time scale. Thus, our algorithm can optimize over a longer time scale. In our QoE metric, we consider the rate variation across the chunks, thus, our algorithm minimizes the variation of rate across the chunks. Finally, since all the above papers only optimize the rate in a shorter time scale, if the bandwidth is poor there may be a stall, in contrast, we use the excess bandwidth to fetch the later chunks at a higher rate, thus, our algorithm do not increase the stall time even when the bandwidth is poor. 

\section{System Model}\label{sec:system}
Suppose that the video consists of  $K$ chunks.  Each chunk corresponds to certain temporal segment of the video. Each chunk is of duration $L$ seconds. Hence, the total duration of the video is $KL$ secods. We consider non-real time video streaming, {\it i.e.}, a video has been recorded previously and needs to be streamed to the users as per their requests. However, the algorithms we propose can be extended to the live TV streaming with a minor modification.
Each chunk $k=1,\ldots, K$ has $N$ tiles in total of the same duration\footnote{We do not assume any fixed shape and size of the FoV. However, the entire chunk has to be divided in different number of tiles. Our analysis will go through even if the number of tiles is different across the chunks. }. The tiles constitute sub-areas within a chunk.
A tile has the same duration as the chunk.  

\textbf{Field of View}: A viewing area for chunk $k$ belongs to the set  $\mathcal{V}_k$ which is the set of sub-sets of $l$ tiles among all $N$ tiles. Thus, the field of view (FoV) $V_{f,k}\in \mathcal{V}_k$ consists of $l$ tiles of a chunk. The video player can choose to download tile $i\in \{1,\ldots, N\}$ of chunk $k$ at the rate $R_{i,k}$. Thus, $L\sum_{i}^{N}R_{i,k}$ is the size of the chunk $k$. 

We assume that the video player adapts the rate of the video at the minimum rate among the tiles in the FoV  which are downloaded. \begin{color}{black}Thus, even though the tiles are downloaded at different rates, the viewer would not like to see spatial quality differences within the FoV. Thus, determining the QoE by the minimum of playback quality of the tiles in the FoV will lead to lower quality variations. Further, SVC encoding allows tiles to be played at lower quality than that downloaded, and thus all the tiles in the FoV can be played at the minimum of the quality of the tiles in FoV.\end{color}

\begin{color}{black}Seeing the FoV with different tiles at different qualities may not be perceived as high QoE to the user. As an example, imagine that half screen is at low quality and half at high quality, that will make the difference of quality in a frame very visual. The differing quality in a chunk would create spots which might cause it to be visually less appealing. Thus, we believe that the streaming service may want to display spatially all the tiles in FoV in the same quality. In the traditional video streaming, the different segments of the chunk are the same quality and the quality is not reduced for parts spatially. Thus, it makes intuitive sense to not have significant spatial quality variations.\end{color}

\begin{color}{black}
We note that even though this assumption of not having spatial quality variation have not been validated in the past as a quality metric, one knows that the spatial quality variations will hurt the perceived quality. We assume that a huge penalty for the spatial quality variations thus aiming to improve the minimum of the spatial quality in the FoV. 
\end{color} 


\textbf{Encoding Scheme}: 
We consider that each chunk is encoded in one of the predefined quality layers. Each layer corresponds to a bit-rate.
The bit-rate $R_{j}$ denotes the rate at the $j$-th layer. Thus, the tile $i$ within a chunk $k$ can be downloaded at rate  $R_{i,k}\in \{R_0,R_1,\ldots, R_{m}\}$ where $R_{m}>\ldots>R_1>R_0$. Hence, the set of possible rates is the same for each tile. 

Note that if the FoV coincides with the tiles where those tiles have been downloaded at zero rate, the user may see black spots. In order to avoid that  we assume that the video player has to download a tile at least at its minimum rate $R_0$. Thus, even if the FoV prediction is erroneous, the user will not watch any black spots. Note that in contrast, the most literature \cite{hosseini, ganainy, Graf, Petrangeli, LeFeuvre, feng, bao} considered to download tiles only within the predicted FoV. Thus, if the prediction is erroneous, the user will  see black spots. 

Our approach can work both with the Adaptive Video Coding (AVC) \cite{avc}  or Scalable Video Coding (SVC) \cite{svc}.
 The main difference between AVC and SVC is that in AVC, each video chunk is stored into independent encoding versions while in SVC the encoding of one version may depend on that of another version. AVC is easy to implement and thus, it is the most popular one. However, SVC's unique encoding scheme has many advantages over the other state-of-the-art schemes.  For example,  if the chunk in AVC is not fully downloaded before its playback deadline, there will be a stall. However, this can be easily mitigated in the SVC: if a chunk can not be downloaded at layer $j+1$, it still can be played without any stall for layer up to $j$. {\color{black} We note that with AVC encoding, the tile can not be played at lower quality than the one fetched, but the quality metric will try to improve minimum quality in FoV thus helping reduce spatial quality variation in the viewed tiles of a chunk. }



\textbf{Download time and Bandwidth}: Different chunks of video are downloaded sequentially into a {\em playback buffer} which contains downloaded but as yet unwatched video. Let $B(t)\in [0,B_{max}]$ denote the buffer occupancy at time $t$, {\it i.e}., the play-time of the video left in the buffer at time $t$.  Chunk $k+1$ can be started to download once the chunk $k$ is downloaded. Once the chunk $k$ is downloaded, the video player waits a time $\zeta_k$ to start downloading the chunk $k
+1$. Specifically, $\zeta_k$ is positive only when the buffer is full, otherwise, it is $0$. Let $t_k$ be the start time of downloading the chunk $k$. We have:
\begin{align}\label{eq:tk}
t_{k+1}=t_{k}+\dfrac{L\sum_{i}R_{i,k}}{C_k}+\zeta_k,\quad t_0=0.
\end{align}
 where $C_k$ is the average bandwidth while downloading the chunk $k$, {\it i.e.}, in the interval $[t_k,t_{k+1}]$. We assume that $C_k$ can be estimated which we explain in detail in Section 4.2.

\textbf{Play time of a chunk}: The buffer occupancy evolves as new chunks are downloaded. When a chunk is downloaded the buffer occupancy increases by $L$ and when the chunk is played, it decreases by $L$. We denote the play-time of the $k$-th chunk as $\tilde{t}_k$, i.e., the $\tilde{t}_k$ is the time when the $k$-th chunk starts playing. If the chunk $k$ is not downloaded before its designated time $\tilde{t}_{k-1}+L$, it has to wait till chunk $k$ is downloaded. The time at which chunk $k$ is downloaded is given by $t_{k-1}+\dfrac{L\sum_{i}R_{i,k}}{C_k}$. 

Hence, for chunk $k>1$, the play time is:
\begin{align}\label{b_k}
\tilde{t}_{k}=\max\left\{\tilde{t}_{k-1}+L,t_{k-1}+\dfrac{L\sum_{i}R_{i,k}}{C_k}\right\}, \quad \tilde{t}_1=t_{\mathrm{ini}}.
\end{align}
where $t_{\mathrm{ini}}$ is the initial start-up time or initial stall time. This is the time the first chunk starts playing. The initial start-up time is often considered to be constant.

\textbf{Stall time}: Note that each chunk constitutes $L$-second worth of video content. Hence, if $\tilde{t}_{k}>(k-1)L+t_{\mathrm{ini}}$, then, there will be stall or re-buffering. Thus the total stall time is $(\tilde{t}_K-(K-1)L-t_{\mathrm{ini}})^{+}$.  Note from (\ref{b_k}) that $\tilde{t}_{K}\geq (K-1)L+t_{\mathrm{ini}}$.  Thus,
\begin{align}\label{stall}
(\tilde{t}_K-(K-1)L-t_{\mathrm{ini}})^{+}=\tilde{t}_K-(K-1)L-t_{\mathrm{ini}}.
\end{align}

\textbf{Maximum Buffer Occupancy}: As mentioned before we assume that as soon as chunk $k$ is downloaded, the video player starts downloading the chunk $k+1$. Hence, $\zeta_k$ is $0$ when the buffer is not full. However one exception is that when the buffer is full, the player waits for the buffer occupancy level to decrease so that the next chunk can be downloaded.

Without loss of generality, we assume that the maximum buffer occupancy is a multiple of $L$.\footnote{It is straight forward to extend to the setting when the maximum buffer occupancy is not a multiple of $L$. The constraints will change, however, the set of constraints will still constitute a convex set. Hence, the analysis will go through.} The buffer can store at most $B$ chunks. Hence, we must have
\begin{align}\label{eq:max_buffer}
t_{B+k}\geq \tilde{t}_k\quad k=1,\ldots, K-B.
\end{align}
$\zeta_k$ is adjusted in (\ref{eq:tk}) such that the above constraint is satisfied. Currently, the maximum buffer occupancy $B$ can be very high, and, therefore, the above constraint is almost always satisfied. 

\textbf{User's utility}:
The user obtains an utility $U(\cdot)$ depending on the rate at which the video is being played. For example, if the chunk is played at rate $R$, the user's utility is $U(R)$ which denotes the user's satisfaction for getting the chunk at rate $R$. $U(\cdot)$ is a strictly increasing function as the user strictly prefers a higher rate.
The play-back rate for a chunk in the FoV is governed by the lowest rate amongst the tiles within the FoV since the video is played at the lowest rate among all the tiles within the FoV.

We assume that
\begin{assumption}
$U(\cdot)$ is a concave strictly-increasing function.
\end{assumption}

The intuition behind the above assumption is that user's QoE increases very fast initially with the play back rate of the video, however, the user becomes indifferent when the play back rate exceeds a certain threshold. 
This is a standard assumption as the user utilities are often assumed to be concave (e.g., smart grid \cite{low2}, quality of service for multimedia \cite{wang}). $U(x)=x^{\alpha}$ where $\alpha\leq 1$ is an example of a concave function.

\section{Problem Formulation and Equivalent Representation}\label{sec:robust}%
We, first, provide the QoE metric of a user. The QoE is a random variable since the FoV, and the bandwidths are random. Thus, we represent a QoE metric which is robust. We, finally, formulate the robust optimization problem. 
\subsection{QoE Representation}
Let $\mathcal{F}_k$ be the set of tiles within the FoV at chunk $k$. Recall that to avoid spatial quality variation, all tiles in the FoV of chunk $k$ should  be played at quality $\min_{i\in \mathcal {F}_k}R_{i,k}$.
Hence, the  utility is governed by the minimum rate among those tiles within the set $\mathcal{F}_k$.
Thus, the video player  wants to maximize $\sum_{k=1}^{K}U(\min_{i\in \mathcal{F}_k}R_{i,k})$.

The video provider also wants to reduce the quality variation across the chunks with a high probability. Hence, it wants to minimize the quality variation among the tiles within the FoV across two consecutive chunks. In other words the video service provider wants to minimize
\begin{align}\label{eq:diff1}
\sum_{k=1}^{K-1}|\min_{i\in \mathcal{F}_{k}}R_{i,k}-\min_{i\in \mathcal{F}_{k+1}}R_{i,k+1}|
\end{align}
QoE also decreases as the stall time increases, hence, the QoE metric is given by
\begin{align}\label{eq:qoe}
\text{QoE}= & \sum_{k=1}^{K}U(\min_{i\in \mathcal{F}_{k}}R_{i,k})-\lambda(\tilde{t}_K-(K-1)L-t_{\mathrm{ini}})\nonumber\\ & -\eta\sum_{k=1}^{K-1}|\min_{i\in \mathcal{F}_{k}}R_{i,k}-\min_{i\in \mathcal{F}_{k+1}}R_{i,k+1}|
\end{align}
$\lambda$ is the weight factor corresponding to the stall time (cf.(\ref{stall})) and $\eta$ is the weight corresponding to the quality variation across the chunks. 
\begin{remark}
The optimal weights ($\lambda, \eta$) corresponding to different factors of the QoE metric are set depending on the video player's objective. In general those values are often obtained from the user's sensitivity analysis of the stall time duration and the rate-variation across the chunks. 
\end{remark}

\subsection{Robust QoE metric and Problem Formulation}
In the previous subsection, we represent the QoE when the FoV of every chunk is known apriori. However, the FoV evolves randomly with time depending on the head movement of the user. 
We consider that the video service provider wants to provide a rate $r$ to the users such that the play back rate will be greater than or equal to that rate $r$ with a high probability. 
In any kind of entertainment or video streaming, the tail probability is important. If the video quality is poor in even for a small portion of the chunk, it will render poor rating, even though the mean quality is good. Further, the video application would like to serve the FoV using same quality tiles so as not to have quality variations in a chunk affecting QoE. Hence, we need to achieve a trade-off such that the viewing quality is good at least with a certain probability $\alpha$.
We, thus, consider a QoE metric where the utility of the user only depends on the rate at least with which the chunk can be viewable with probability $\alpha$.  
Note that such metric is already proposed in the communication channel, or the electricity market.  We, now, formulate a QoE metric, and subsequently, propose solution to maximize the above metric.

We introduce a notation which we use throughout this section.
\begin{definition}
Let $\mathcal{A}_{\alpha,k}$ be the set of tiles of chunk $k$ which has the property that the FoV is a subset of $\mathcal{A}_{\alpha,k}$ with probability at least $\alpha$ for chunk $k$. If there are multiple sets which satisfy the above condition, then, $\mathcal{A}_{\alpha,k}$ is considered to be the set with the lowest cardinality.\footnote{Cardinality of the set is the number of elements in the set.}
Let $\mathcal{A}^{C}_{\alpha,k}$ be the complement of the set $\mathcal{A}_{\alpha,k}$.
\end{definition}
We do not assume that the exact FoV is known beforehand. However, we do assume that the statistics are known. The statistics can be obtained from the crowd sourced data, i.e., FoV distributions different users watching the same video.
Also note that in our algorithm we do not need to know the exact distribution, mean, or variance of the FoV. We need to build the set $\mathcal{A}_{\alpha,k}$, the set of tiles which can be viewed with probability $\alpha$.\footnote{Our analysis will go through even if $\alpha$ is different across the chunks.}
%
%

Our QoE metric  is specified for robust optimization technique. We only need to compute the tiles which constitute the set $\mathcal{A}_{\alpha,k}$, the tiles which constitute the set of tiles can be seen with a high probability. 

$\mathcal{A}_{\alpha,k}$ is the set of the tiles of the lowest cardinality which specifies that the FoV will be a subset of this set with probability $\alpha$. Thus, the lowest rate among the tiles of $\mathcal{A}_{\alpha,k}$ gives the lowest rate the user will observe the video with the probability $\alpha$. The video content provider wants to provide a  rate a user such that the play-back rate will be lower that rate only with probability $1-\alpha$.

The  utility is considered to be governed by the minimum rate among those tiles within the set $\mathcal{A}_{\alpha,k}$.
Thus, the video player  wants to maximize $\sum_{k=1}^{K}U(\min_{i\in \mathcal{A}_{\alpha,k}}R_{i,k})$. 

\begin{remark}
	$\alpha$ is a parameter which needs to be determined from the market research. In stochastic optimization problems, the robustness is an important issue where the optimizer wants to protect itself from the error in data or high variance. However, a too-high $\alpha$ value may give very pessimistic solution. For robust optimization problems, $\alpha$ is generally taken to be in the interval $[0.9, 0.99]$ \cite{alpha}.
\end{remark}

As stated in Section 4.1, the QoE metric decreases as the rate varies across the chunks (cf. (\ref{eq:diff1})). Hence, the video player wants to minimize the quality variation among the tiles within the set $\mathcal{A}_{\alpha,k}$ across two consecutive chunks. In other words the video service provider wants to minimize 
\begin{align}\label{eq:diff}
\sum_{k=1}^{K-1}|\min_{i\in \mathcal{A}_{\alpha,k}}R_{i,k}-\min_{i\in \mathcal{A}_{\alpha,k+1}}R_{i,k+1}|
\end{align}
Recall from (\ref{eq:qoe}) that QoE also decreases as the stall time increases. Also note from Section 3 that  we assume each tile of every chunk $k$ has to be downloaded at the minimum rate even those tiles are not in the set $\mathcal{A}_{\alpha,k}$ in order to avoid the situation where the user will see block spots when the FoV does not coincide with the exact prediction. Thus, the {\em robust} QoE metric is given by
\begin{align}\label{eq:probust}
\text{QoE}= & \sum_{k=1}^{K}U(\min_{i\in \mathcal{A}_{\alpha,k}}R_{i,k})-\lambda(\tilde{t}_K-(K-1)L-t_{\mathrm{ini}})\nonumber\\ & -\eta\sum_{k=1}^{K-1}|\min_{i\in \mathcal{A}_{\alpha,k}}R_{i,k}-\min_{i\in \mathcal{A}_{\alpha,k+1}}R_{i,k+1}|
\end{align}


Thus, formally, the optimization problem that the video player  is solving is given by
\begin{eqnarray}
\mathcal{P}:\text {maximize } &\sum_{k=1}^{K}U(\gamma_k)-\lambda(\tilde{t}_K-(K-1)L-t_{\mathrm{ini}})-\nonumber\\
&\eta\sum_{k=1}^{K-1}|\gamma_{k}-\gamma_{k+1}|\label{eq:robust}\\
\text{subject to } &  (\ref{b_k}), (\ref{eq:tk}), (\ref{eq:max_buffer})\nonumber\\
& \gamma_k= \min_{i\in \mathcal{A}_{\alpha,k}}R_{i,k} \label{eq:constraint}\\
 \text{var}: & R_{i,k}\in \{R_0,\ldots, R_{m}\}\nonumber
\end{eqnarray}
Note from (\ref{eq:constraint}) that $\gamma_k$ denotes the minimum rate of the tiles within the set $\mathcal{A}_{\alpha,k}$.

\begin{remark}
 We assume the Constant Bit-Rate (CBR) technique where the spatial encoding complexity is constant through the video.  Our analysis will go through even if the encoding complexity is different for different frames. 
\end{remark}
\subsection{Equivalent Representation}
The constraint set in the optimization problem $\mathcal{P}$ is not in the convex form because of the constraints in (\ref{eq:constraint}) and the discrete strategy space. In the following, we provide equivalent representations of the constraints (\ref{eq:constraint}) which are in convex form. 

Note that in an optimal solution of $\mathcal{P}$ only the tiles within the set $\mathcal{A}_{\alpha,k}$ should be fetched at rates higher than $R_0$ if possible. The rate at which the tiles in the set $\mathcal{A}^{C}_{\alpha,k}$ are fetched do not contribute to the utility.  Hence, the tiles in the set $\mathcal{A}^{C}_{\alpha,k}$ is fetched at  rate $R_0$ in an optimal solution. Recall that  $\mathcal{A}^{C}_{\alpha,k}$ is the complement of the set $\mathcal{A}_{\alpha,k}$.

 If the tiles within $\mathcal{A}_{\alpha,k}$ are not fetched in the same quality, then, the utility is only governed by the minimum among the rates at which the tiles are fetched. Hence, in an optimal solution one should fetch all the tiles in $\mathcal{A}_{\alpha,k}$ at the same rate.  Now, we propose an equivalent representation of the optimization problem.

\begin{theorem}\label{thm:equi}
Suppose that $R^{*}_{i,k}$ is the optimal solution of $\mathcal{P}$. Now, consider the following solution: $R^{\prime}_{i,k}=\min_{i\in \mathcal{A}_{\alpha,k}}R^{*}_{i,k}$ $\forall i\in \mathcal{A}_{\alpha,k}$ and $R^{\prime}_{i,k}=R_0$ $\forall i\in \mathcal{A}^{C}_{\alpha,k}$. Then, $R^{\prime}_{i,k}$ is also an optimal solution.
\end{theorem}
The above theorem will assist us to represent $\mathcal{P}$ in an equivalent form with $\gamma_k$ as decision variables. Note that the tiles within the set $\mathcal{A}_{\alpha,k}$ $k=1,\ldots,K$  are fetched at rate $\gamma_k$ and the other tiles at rate $R_0$, then, the download time in (\ref{eq:tk}) is modified as
\begin{align}\label{eq:downtime}
t_k=t_{k-1}+\dfrac{L(|\mathcal{A}_{\alpha,k}|\gamma_k+|\mathcal{A}^{C}_{\alpha,k}|R_0)}{C_k}+\zeta_k,\quad t_0=0.
\end{align}
where $|\cdot|$ denotes the cardinality of the set.

Finally, we represent the constraint in (\ref{b_k}) as the following:
\begin{align}\label{eq:equi}
& \tilde{t}_k\geq \tilde{t}_{k-1}+L,  \tilde{t}_1\geq t_{\mathrm{ini}} \quad\tilde{t}_k\geq \dfrac{L(|\mathcal{A}_{\alpha,k}|\gamma_k+|\mathcal{A}^{C}_{\alpha,k}|R_0)}{C_k}.
\end{align}
Now, we are ready to represent $\mathcal{P}$ in an equivalent form.
\begin{proposition}
Consider the following optimization problem
\begin{eqnarray}
& \mathcal{P}_{\mathrm{eq}}: \text {maximize } & (\ref{eq:robust})\nonumber\\
& \text{subject to } &  (\ref{eq:equi}), (\ref{eq:downtime}), (\ref{eq:max_buffer})\nonumber\\
& \text{var}: & \gamma_k\in \{R_0,\ldots, R_m\}\nonumber
\end{eqnarray}
The optimization problem $\mathcal{P}_{\mathrm{eq}}$ is equivalent to (\ref{eq:robust}), and is a discrete optimization problem.
\end{proposition}

\begin{remark}
The optimization problem $\mathcal{P}_{\mathrm{eq}}$ only gives $\gamma_k$. However, we need the rates $R_{i,k}$. We obtain that in the following manner: if $\gamma_k$ is the optimal solution of $\mathcal{P}_{\mathrm{eq}}$, then we can obtain the optimal rates as follows: $R_{i,k}=\gamma_k$ for all $i\in \mathcal{A}_{\alpha,k}$ and $R_{i,k}=R_0$ for all $i\in \mathcal{A}^{C}_{\alpha,k}$.
\end{remark}

\section{Algorithms and Results}\label{sec:result}
The problem in $\mathcal{P}_{\mathrm{eq}}$ is not a convex problem. In order to make it convex, we first relax the discrete strategy space. Subsequently, we state an algorithm (\textbf{360-ROBUST}) which computes a feasible solution from the optimal solution from the relaxed problem. Finally, we state the online version of the algorithm 360-ROBUST which adapts to the real time variation of the bandwidth prediction, and the FoV prediction. 
\subsection{Relaxation}




The relaxed problem is stated in the following:
\begin{eqnarray}
\mathcal{P}^{\mathrm{relaxed}}_{\mathrm{eq}}: & \text{ maximize } (\ref{eq:robust})\nonumber\\
& \text{subject to } &  (\ref{eq:equi}), (\ref{eq:downtime}), (\ref{eq:max_buffer}),  R_{0}\leq \gamma_k\leq R_m\nonumber\\
& \text{var}: & \gamma_k\nonumber
\end{eqnarray}
Note that $\gamma_k$ is now continuous an can take any value within the interval $[R_{0},R_m]$ unlike the discrete values in 
$\mathcal{P}_{\mathrm{eq}}$. 

\begin{proposition}\label{prop:convexrobust}
$\mathcal{P}^{\mathrm{relaxed}}_{\mathrm{eq}}$  is a convex optimization problem.

The optimal value of the objective function in $\mathcal{P}^{\mathrm{relaxed}}_{\mathrm{eq}}$  is greater than or equal to the optimal value of the objective function in $\mathcal{P}_{\mathrm{eq}}$. 
\end{proposition}
Since the relaxed problem $\mathcal{P}^{\mathrm{relaxed}}_{\mathrm{eq}}$ is convex, we can solve it using the standard convex optimization techniques efficiently. However, we have to transform the optimal solution of $\mathcal{P}^{\mathrm{relaxed}}_{\mathrm{eq}}$ into a feasible solution as the optimal solution $R^{*}_{i,k}$ of $\mathcal{P}^{\mathrm{relaxed}}_{\mathrm{eq}}$  may not belong to the discrete set $\{R_0, . . . , R_m\}$. Since the strategy space is relaxed, the optimal value will be greater than or equal to the optimal value of the exact problem. The equality arises only when the solution of the relaxed problem is a feasible solution of the exact problem. In the following section, we provide heuristic solutions to obtain the feasible solutions from the optimal solution of the relaxed problem.
\subsection{Heuristic to obtain the feasible solution}
Let $\gamma^{*}_k$ be the optimal solution of the relaxed problem $\mathcal{P}^{\mathrm{relaxed}}_{\mathrm{eq}}$ for all $k$. $\gamma_{k}^{*}$ can be obtained using the convex optimization tools as the problem is convex (Proposition~\ref{prop:convexrobust}). Now, we describe a simple heuristic to obtain a feasible solution based on the optimal solution.
\begin{figure}
	\vspace{-.1in}
\begin{algorithm}[H]
\small
\begin{algorithmic}
\STATE {\bf Inputs}: $\mathcal{A}_{\alpha,k}$, $C_k$, the optimal solution $\gamma^{*}_k$ of $\mathcal{P}^{\mathrm{relaxed}}_{\mathrm{eq}}$, and the set  $\{R_0,\ldots, R_m\}$.
\STATE {\bf Initialization:} $L_0=0$.
\FOR{$k=1,\ldots, K$}
\STATE $\gamma_k=\max_{j\in\{1,\ldots,m\}} \{R_{j}: R_{j}\leq \gamma^{*}_k\}$.
\STATE $\bar{\gamma}_{k}=\min_{j\in \{1,\ldots,m\}} \{R_{j}: R_{j}\geq \gamma^{*}_{k}\}$.
\STATE Compute $L_k=L_{k-1}+(\gamma^{*}_{k}-\gamma_{k})$.
\ENDFOR
\IF{$L_{K}<\bar{\gamma}_{k}-\gamma_k$}
\STATE Exit since it is the optimal solution.
\ENDIF
\FOR{$k=K, K-1, \ldots , 2$}
\IF{$L_k\geq\bar{\gamma}_{k}-\gamma_k$}

\STATE $L_{k-1}=L_{k}-(\bar{\gamma}_{k}-\gamma_k)$.
\STATE Update: $\gamma_k=\bar{\gamma}_{k}$ for all $i\in \mathcal{A}_{\alpha,k}$.
\ENDIF
\ENDFOR
\STATE {\bf OUTPUT} $\gamma_k$.
\end{algorithmic}
\caption{: Algorithm \textbf{360-ROBUST} provides a feasible solution of $\mathcal{P}_{\mathrm{eq}}$ from the optimal solution of $\mathcal{P}^{\mathrm{relaxed}}_{\mathrm{eq}}$}\label{algo2} 
\end{algorithm}
\vspace{-0.2in}
\end{figure}

As described in Algorithm~\textbf{360-ROBUST}, our high-level idea is to
 set $\gamma_k$ at the maximum possible rate which is less than or equal to the optimal solution $\gamma_k^{*}$ of the relaxed problem. Hence, $\gamma_k$ is obtained by simply down-quantizing the optimal rate obtained from the relaxed problem. The algorithm proceeds to compute the bandwidth saved compared to the solution of the relaxed problem because of the down-quantization ($L_k$).


 The algorithm saves any additional bandwidth for down-quantization to download some chunks at the immediate higher level if it does not increase the stall duration. Specifically, the algorithm fetches the chunk $K, K-1, \ldots$ at the immediate next higher level and so on until the stall time does not increase ($L_k$ remains positive).  
 
A reader may ask the question why not assigning the additional bandwidth to fetch the earlier chunk at  a higher rate rather fetching the later chunks at a higher rate. The reasons are the following. First, this will help reduce the number of switches of rate among two consecutive chunks since starting from last chunk, it is more likely to increase the quality of consecutive chunks. However, due to earlier deadlines of the earlier chunks, increasing quality of consecutive earlier chunks is much less likely.   
Second, there may be an error in estimating the bandwidth; if the estimated bandwidth is erroneous, there may not be any bandwidth to fetch the future chunks even in the lower rate, if the additional bandwidth is used up in downloading the current chunk at the immediate higher level.


\begin{theorem}
The algorithm 360-ROBUST gives a feasible solution of the problem $\mathcal{P}_{\mathrm{eq}}$. The stall time does not increase compared to the optimal solution of the relaxed problem.
\end{theorem}
Note that if the quantization levels are very close, then the feasible solution obtained will be very close to the optimal solution. Also note that Algorithm 360-ROBUST scales linearly with the number of chunks. 
We, now, compute the optimality gap of Algorithm 360-ROBUST.
\begin{theorem}\label{thm:opt_gap1}
Suppose the optimal value of the objective function (QoE) in $\mathcal{P}^{\mathrm{relaxed}}_{\mathrm{eq}}$ is $\mathrm{OPT}$.
The value of the objective  function attained by Algorithm 360-ROBUST is at least
\begin{align}
\mathrm{OPT}-\max_{j\in {0,\ldots,m-1}}K(U(R_{j}+\dfrac{R_{j+1}-R_{j}}{K})-U(R_{j})).
\end{align}
\end{theorem}

\begin{proof}
	Note that if there are enough extra bandwidth the algorithm fetches the tiles in the next higher rate. Suppose $R_{i,k}^{*}$ be the optimal rate for tile $i$ given by the optimal solution of the relaxed problem $\mathcal{P}^{\mathrm{relaxed}}_{\mathrm{eq}}$ for chunk $k$. Suppose that $R_{i,k}$ be the rate given by simple down-quantization where $R_{i,k}$ is defined in the following
	\begin{align}
	R_{i,k}=\max_{j\in\{0,\ldots,m\}}R_{j}\leq R_{i,k}
	\end{align}
	Since we are down-quantizing, thus, the QoE bound is off by at most
	\begin{align}
	\min_{i\in \mathcal{A}_{\alpha,k}}U(R^{*}_{i,k})-U(R_{i,k})
	\end{align}
	
	Now, note that since the extra bandwidth is used to fetch the tiles in the next higher rate, if $R_{i+1,K}-R_{i,K}\leq L_{K} $ then the algorithm fetches the tiles in the next higher rate $R_{i+1,K}$ for chunk numbered $K$. $L_{K-1}=L_{K}-R_{i+1,K}+R_{i,K}$, if $L_{K-1}\geq R_{i+1,K-1}-R_{i,K}$, then the algorithm fetches the tiles at the next higher rate for chunk numbered $K-1$. 
	
	Thus, our proposed algorithm fetches total rate which is off from the optimal one by at most the following amount:
	\begin{align}
	\sum_{k=1}^{K}R_{i,k}^{*}-R_{i,k}\leq \max_{j\in \{0,\ldots,m-1\}}(R_{j+1}-R_{j}).
	\end{align}
	
	Since $U(\cdot)$ is a concave function, thus, $KU(\sum_{i=1}^{K}x_i/K)\geq\sum_{i=1}^{K}U(x_i)$. Now, the total amount of rate downloaded using our proposed algorithm is upper bounded by at most
	\begin{align}
	&\sum_{k=1}^{K}U(R^{*}_{i,k})-U(R_{i,k})\nonumber\\
	& \leq K[U(R_{i,k}+\sum_{k=1}^{K}(R^{*}_{i,k}-R_{i,k})/K)-U(R_{i,k})]\nonumber\\
	& \leq \max_{j\in \{0,\ldots,m-1\}}K [U(R_j+(R_{j+1}-R_{j})/K)-U(R_j)]
	\end{align}
	Hence, the result follows.
\end{proof}


The above theorem provides the bound on the QoE gap attained by the Algorithm 360-ROBUST as compared to the optimal value.  Note from the second part of Proposition~\ref{prop:convexrobust} that the optimal value of the relaxed version $\mathcal{P}^{\mathrm{relaxed}}_{\mathrm{eq}}$ is greater than or equal to the optimal value of the exact  one $\mathcal{P}_{\mathrm{eq}}$. Thus, the difference between the  value of the objective function provided by Algorithm 360-ROBUST and that of the optimal value of the exact problem $\mathcal{P}_{\mathrm{eq}}$ is {\em at most} 
\begin{align}
\max_{j\in {0,\ldots,m-1}}K(U(R_{j}+\dfrac{R_{j+1}-R_{j}}{K})-U(R_{j})).
\end{align}

 The above also shows that the optimal QoE gap is small when the quantization gaps are small, {\em i.e.}, the difference between $R_{j+1}$ and $R_{j}$ are small for all $j\in 0,\ldots, m-1$.  Note that if the utility function is linear, the optimality gap is independent of the number of chunks. Specifically:
\begin{corollary}
If $U(\cdot)$ is linear ($U(x)=x$), the value attained by Algorithm 360-ROBUST is at least
\begin{align}
\mathrm{OPT}-\max_{j\in {0,\ldots,m-1}}(R_{j+1}-R_{j})
\end{align}
\end{corollary}


\subsection{Online Algorithm}
Based on the offline algorithm, we now propose the online algorithm which will adapt to the variation of the bandwidth and FoV dynamically, both of which cannot be perfectly predicted, at runtime. In the online version, the algorithm will decide the rate at which tiles will be downloaded for the immediate next chunk. We consider Receding Horizon Control (RHC) \cite{adam} or {\em sliding window} control type algorithm with a forward-looking window size of $W$ chunks.

Specifically, when the tiles of the $c$-th chunk are being fetched, the online algorithm calculates the download rates $\gamma_k$ for chunks $k=c+1, \ldots,c+W$ by solving the following optimization problem:
\begin{eqnarray}
\mathcal{P}^{\mathrm{online}}_{\mathrm{eq}}: & \text {maximize } &\sum_{k=c+1}^{c+W}U(\gamma_k)-\nonumber\\ & &\lambda(\tilde{t}_{W+c}-(W+c-1)L-t_{\mathrm{c}})\nonumber\\ & &-\eta\sum_{k=c+1}^{c+W}|\gamma_{k}-\gamma_{k-1}|\nonumber\\
& \text{subject to } &  (\ref{eq:equi}), (\ref{eq:downtime}), (\ref{eq:max_buffer})\nonumber\\
& \text{var}: & \gamma_k\in \{R_0,\ldots, R_{m}\}\nonumber
\end{eqnarray}


The time horizon is limited we only want to obtain an optimal solution for $W$ chunks ahead. In $\mathcal{P}^{\mathrm{online}}_{\mathrm{eq}}$, $\gamma_{c}$ is the minimum rate among the tiles of the set $A_{\alpha,c}$ to be downloaded. $t_{c}$ is the time at which the $c$-th chunk will stop downloading. The $c+1$-th chunk can only start downloading at time $t_{c+1}$.
The algorithm runs after each chunk download to determine the quality for the next chunk.
The receding horizon control has advantages compared to the fixed horizon control, as it can adapt to the changes dynamically after each chunk. 

\textbf{The algorithm}: We, now, describe the algorithm. Similar to the offline scenario, we first relax the integer constraint to obtain the convex relaxation version of the problem. Subsequently, we obtain the optimal solution. Finally, we employ Algorithm 360-ROBUST to obtain the feasible solution. We only have to employ the algorithm for $c+W$  chunks starting from chunk $c+1$ instead of running the algorithm from $k=1$ to $K$.

The above online algorithm has sub-linear regret and sub-linear competitive bound \cite{adam}. In the following, we discuss how we predict the bandwidth and the FoV in an online manner. 

\textbf{Bandwidth Prediction Error}:  We download all the tiles of the first $\eta$ chunks at the base layer $R_0$ in order to estimate the bandwidth.  The bandwidth for the future after time $t$ is predicted as the harmonic mean of the observed values for $n$ number of samples immediately preceding  time $t$. A similar approach has been used in previous non-360 video streaming systems \cite{yin,jiang}. Note that the estimated bandwidth is assumed to be constant through out the window $W$. In the simulation, we set $\eta$ at $2$, and $n$ at $200$. 

Note that once the chunk $k$ is being downloaded,  we do not change its rate it in the online approach. However, the bandwidth may be very bad while downloading the chunk $k$ which may result into the stall. We can minimize it by getting all the tiles only in the base layers in the AVC\footnote{If we are using the SVC, then we can ignore the enhancement layers and stop downloading all the enhancement layers.}. On the other hand, if the bandwidth is high compared to the estimated value, we do not increase the download rate, rather we keep downloading the $k$th chunk. Running algorithm to determine fetching for the chunks after $k$ can lead to an increase in the download rates of the future chunks.


\textbf{FoV prediction error}: In the optimization problem we try to fetch the tiles which have high probability to be part of the FoV. Hence, the user may watch the same quality video. Thus, our algorithm is robust in providing the same play-back rate with a high probability. Thus, the impact of the FoV prediction error will be low. 
%
However, a user's FoV may not consist of tiles which are not in the part of $\mathcal{A}_{\alpha,c}$ for the current chunk $c$, {\em the user can still watch the video in the base layer. } Hence, there will be no stall or ``black screen". 

Note that we compute an empirical distribution over the set of FoV from the past-user's experience.  We also consider the current user's head position in order to compute the distribution of the FoV. Specifically, we set an weight $x_k$ ($0\leq x_k\leq 1)$ on the current FoV (FoV for chunk $k-1$) for chunk $k$ and  put an weight $1-x_k$ on the empirical distribution of the FoV for chunk $k$. Since we can not rely on the current head position for too long, we, thus, reduce the weight on the current FoV for future chunks.  Specifically, in the online version, for chunk  $k=c+1,\ldots, c+W$,  we set $x_{k+1}=x_{k}/(k-c-1)$. The set $\mathcal{A}_{\alpha,k}$ is then updated for chunks $k=c+1,\ldots, c+W$ accordingly. Using the updated distribution, we obtain the download rate using Algorithm 360-ROBUST.

\section{Trace-Driven Emulation Results}\label{sec:simulation}

In this section, we empirically evaluate the performance of our algorithm based on the trace based emulation. We, first, explain how the bandwidth and FoV traces are collected and used to estimate in online manner. Subsequently, we explain our emulation set up. Finally, we compare our algorithms compared to two baseline algorithms.  

\subsection{Bandwidth and FoV Traces }

\textbf{Bandwidth Traces}:  For bandwidth traces, we used a dataset from \cite{trace1}, which consists of continuous 10ms measurement of video streaming throughput of a moving device.  We use 40 traces which are at least 240 seconds long. We estimate the bandwidth for time $t$ and beyond by taking the harmonic mean of the $200$ samples till time $t$. 

\textbf{FoV traces}: For FoV traces, we use a setting similar to the one described in \cite{feng}. Specifically, the trace consists of the head movement data of 50 users across the university.
Each user wears a Google Cardboard viewer \cite{cardboard} with
a Samsung Galaxy S5 smartphone placed into the back of
it. The smartphone plays four short YouTube 360 videos
(duration from 1min 40 sec to 3 min 26sec) of different genres. Meanwhile,
a head tracker app runs in the background and sends
the raw, yaw, pitch, and roll readings to a nearby laptop using
UDP over WiFi (latency less than 1ms). During the playback, the users can view at any
direction by freely moving their heads. We first compute the empirical distribution of the user's head movements by considering the data of 40  users. We then use the computed  empirical distribution as the distribution of the FoV for the remaining 10 users to construct the set $\mathcal{A}_{\alpha,k}$ for all $k$. 



We take a 4 minute video, with a chunk duration of 2 seconds. The total number of chunks ($K$) is $120$.  The tile configuration for each chunk is considered to be $4\times 8$; each chunk consists of 32 tiles.   We considered an equi-rectangular projected video chunks. 
Equi-rectangular projection is the most popular one because of its simplicity, and is also used by YouTube.The FoV is considered to be  120 degrees in the horizontal direction and 120 degree in the vertical direction \cite{feng}.

The different tiles are encoded independently in different rates such that each chunk's rate is $x$ Mbps. We consider that the encoded rates belong to  $ \{8,16,24,32\}$. Since there are 32 tiles, each tile's average rate is $x/32$ Mbps which belongs to the  $\{0.25, 0.5, 0.75 ,1\}$. We transmit the video from the server, which is streaming to the client. 


In order to make the decision for the chunk quality, we employ the online algorithm that we have introduced in Section 4.3. 
We run the online algorithm after the completion of the download of the each chunk for $W$ chunks ahead. Recall that $W$ is the size of the sliding window. Though our algorithm can update $W$ independently for each user, we fix  $W$  at the same value for each user. Unless explicitly mentioned, we assume that $W=5$. 
  The utility function is assumed to be linear. Without loss of generality, we consider $U(x)=x$.   Recall from (\ref{eq:robust}) that $\lambda$ is the weight corresponding to the stall time. We assume $\lambda=100$, {\it i.e.}, we give more preference to minimize the stall time.  We set $\alpha$ as $0.95$, $x_k$ as $0.6$. 

\subsection{ Baseline Algorithms}
 We compare the online version of our proposed algorithm 360-ROBUST  with the following baseline algorithms (BA1 and Full). 

\textbf{Algorithm Full}: In Full \cite{Graf}, the tiles within the current viewport is downloaded with the highest possible rate while the rest of the tiles are downloaded with the minimum rate. Thus, though similar to our algorithm, the tiles are downloaded at least at the minimum rate, unlike our algorithm, the tiles only within the current viewport are only downloaded at a higher rate in the Full algorithm. 

\textbf{BA1}: The BA1 algorithm \cite{Petrangeli} is an extension of Full. In the BA1, the tiles adjacent to the ones in the current viewport are also downloaded with a higher quality if there is enough bandwidth.

We study the strength of our proposed algorithm Algorithm~360-ROBUST with respect to the two baseline algorithms. We evaluate the algorithms' achievable QoE (\ref{eq:robust}), the distributions of the bitrates in the FoV, and the impact of the window size $W$ on the overall QoE metrics by taking the average over 10 users. 

\vspace{-.1in}
\subsection{Results and Discussion}
\begin{figure*}
	\vspace{-.1in}
\begin{minipage}{0.23\textwidth}
\includegraphics[trim=0in 0in 0.4in 0in, clip, width=\textwidth]{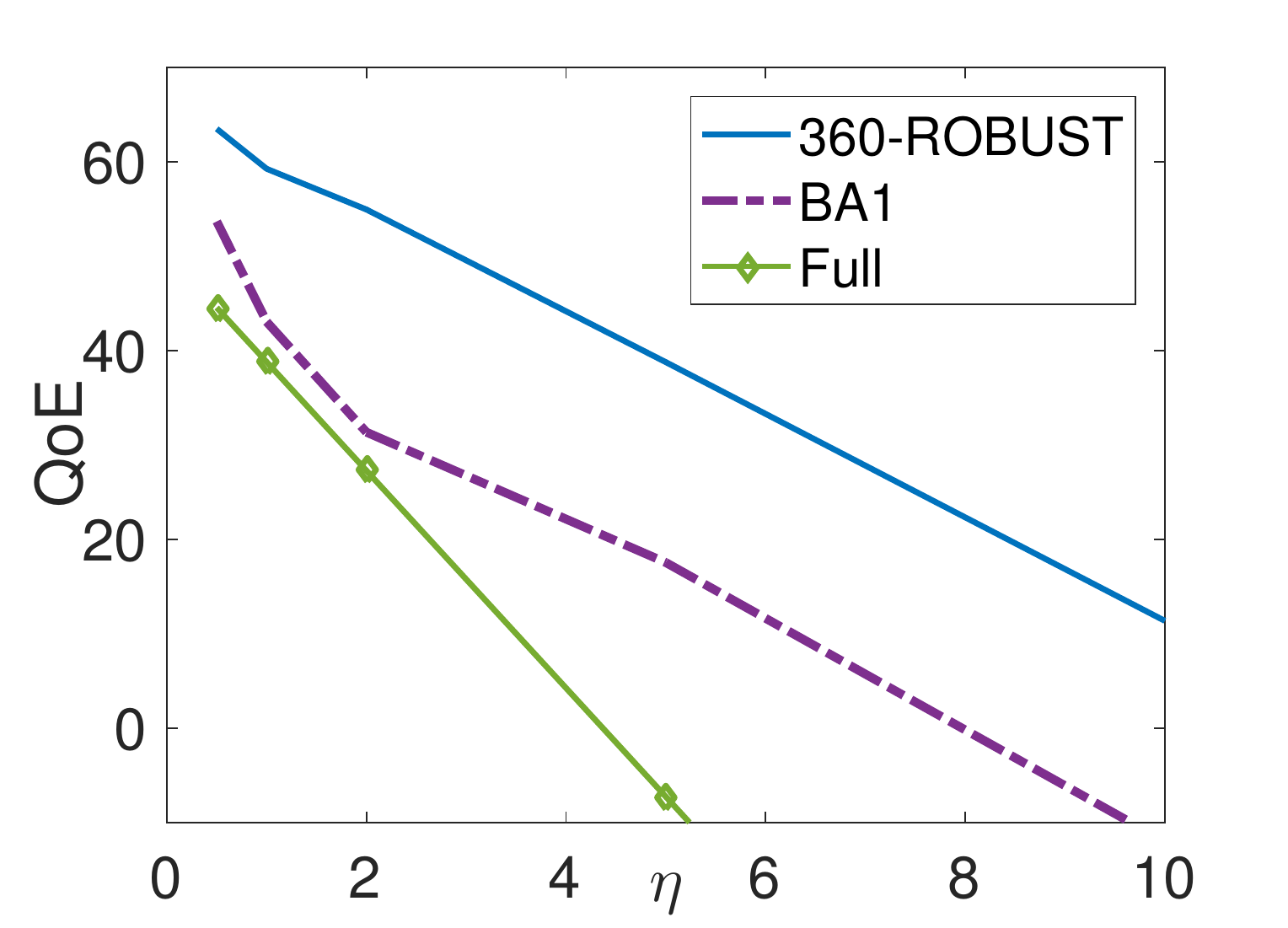}
\caption{ The variation of QoE with $\eta$.}
\label{fig:eta}
\end{minipage}
%
\begin{minipage}{0.23\textwidth}
\includegraphics[trim=0in 0in 0in 0in, clip,width=\textwidth]{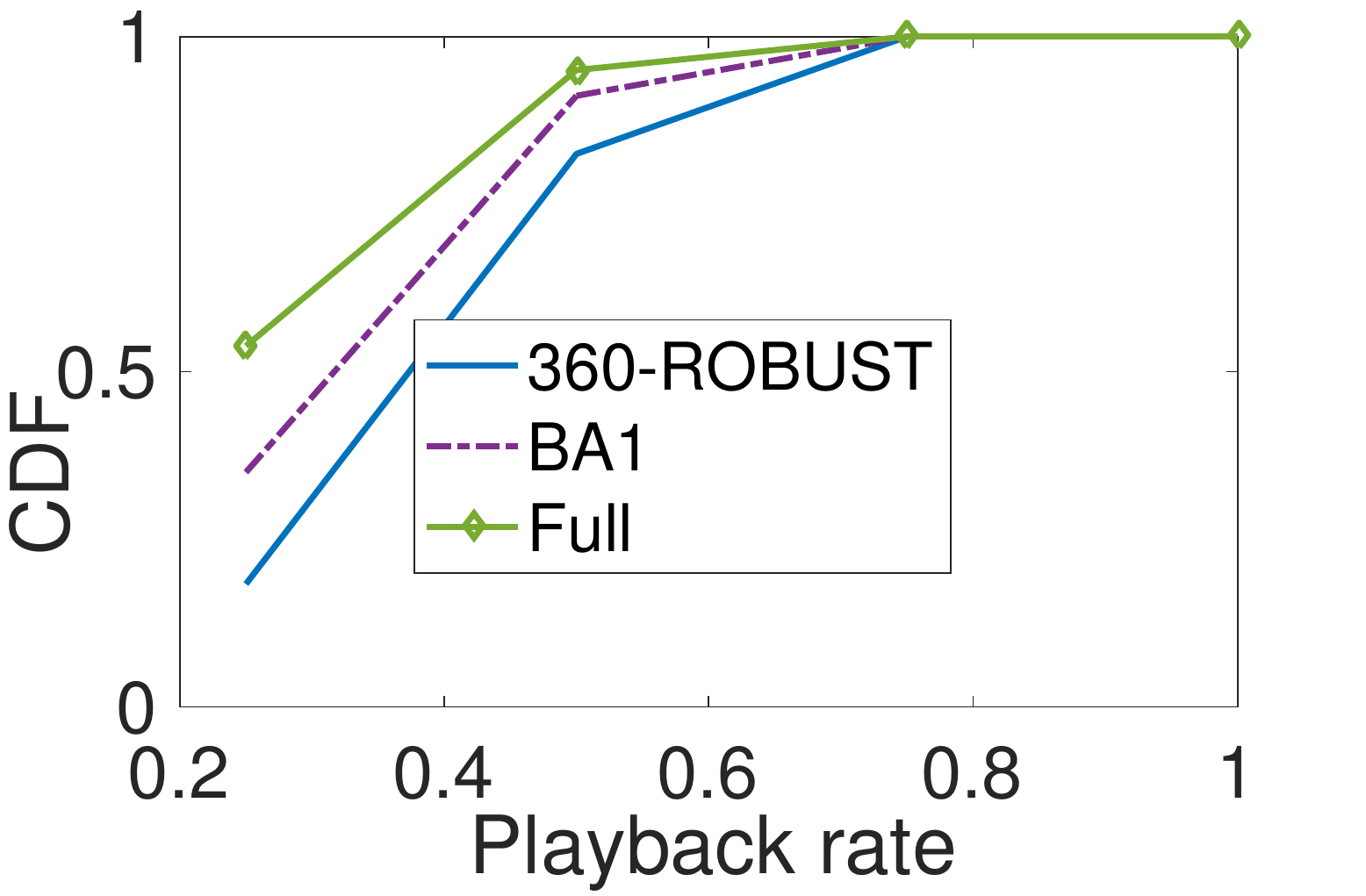}
\vspace{-0.1in}
\caption{Rate distribution of the minimum rates among the tiles within set $\mathcal{A}_{\alpha,k}$ (or, $\gamma_k$) for $\eta=1$.]}
\label{fig:rate}
\end{minipage}
\begin{minipage}{0.23\textwidth}
\includegraphics[trim=.1in 0in 0.1in 0.1in, clip,width=\textwidth]{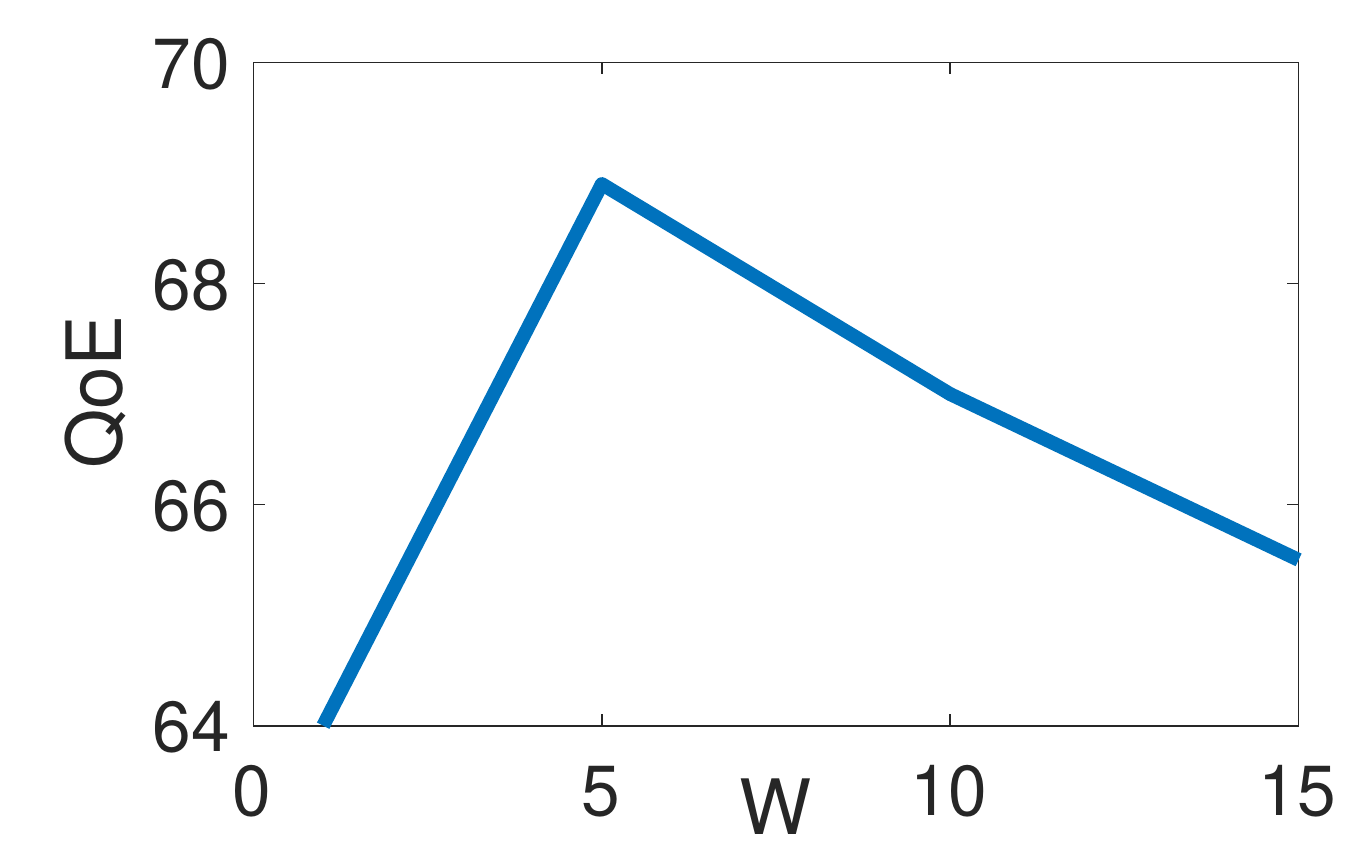}
\caption{{The variation of QoE of 360-ROBUST with the window size $W$ for $\eta=1$.}
}
\label{fig:window}
\end{minipage}
\begin{minipage}{0.23\textwidth}
\includegraphics[trim=.1in 0in 0.6in 0.3in, clip,width=\textwidth]{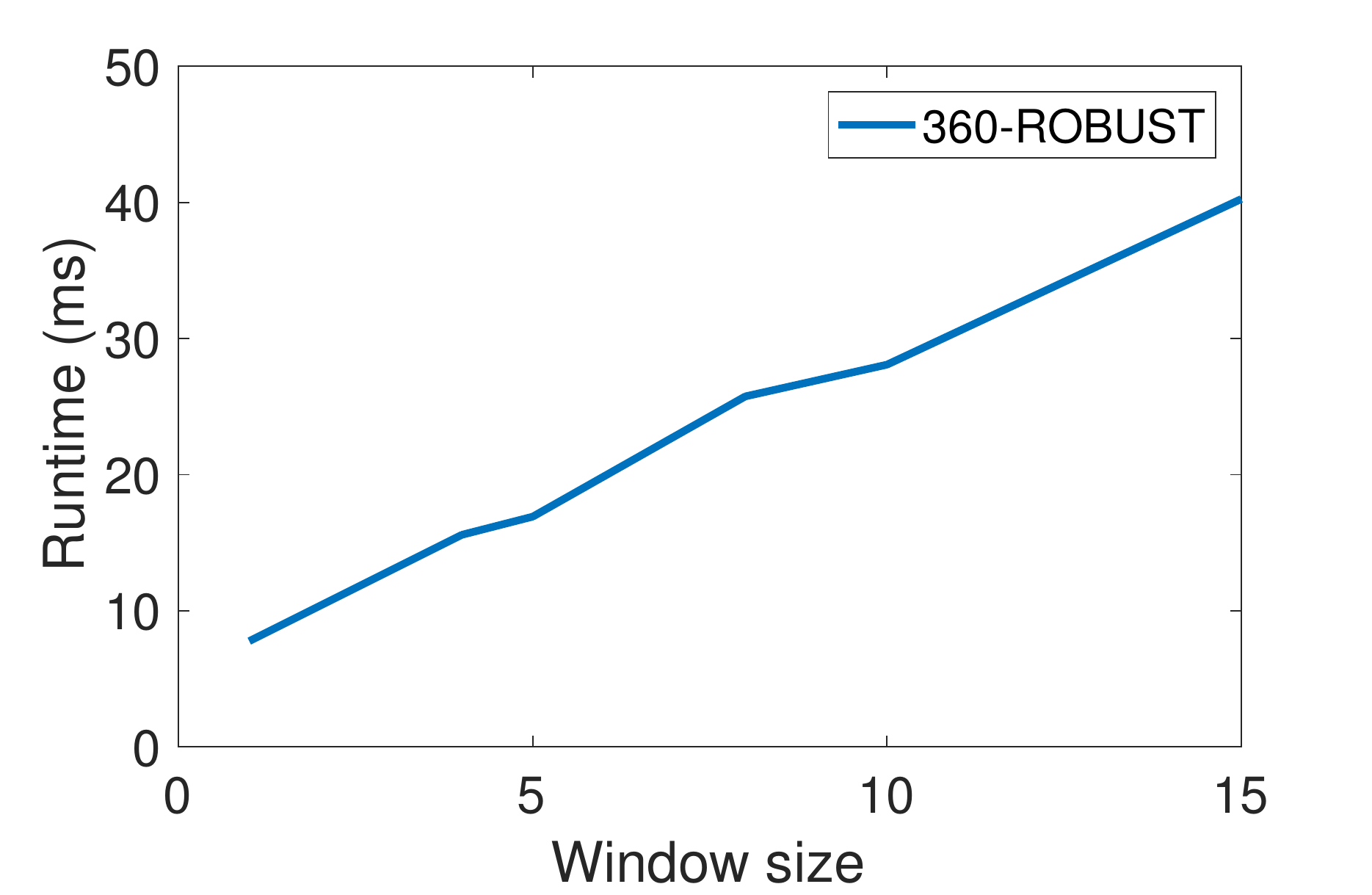}
\caption{Variation of Run-time of Algorithm 360-ROBUST with $W$.}
\label{fig:run_time}
\end{minipage}

\vspace{-0.1in}
\end{figure*}


 \subsubsection{Effect of $\eta$}
Recall from (\ref{eq:robust}) that $\eta$ is the weight corresponding to the quality variation across the chunks. Fig.~\ref{fig:eta} shows that as $\eta$ increases, the QoE decreases as is apparent from (\ref{eq:robust}). Since Algorithm 360-ROBUST is obtained from the optimal solution of $\mathcal{P}^{\mathrm{relaxed}}_{\mathrm{eq}}$, the performance of Algorithm 360-ROBUST would have been optimal if the number of quantization levels would have been large.

Note that the QoEs attained by the baseline algorithms are at least 30\% lower QoE compared to the our proposed algorithms.  The performance is even better when $\eta$ is better. The better performance of our proposed algorithm compared to the baseline algorithm is because of two reasons: First, the baseline algorithms do not take into account of the variation of the rates within the FoV across the chunks. The baseline algorithms only maximize the rates for the immediate next chunk. Second, the baseline algorithms rely on the estimation of the head movement. However, we also use the crowd-sourced data, thus, our proposed algorithm can optimize over $W$ chunks ahead. In contrast, the baseline algorithms only optimize over the next immediate chunk.  Also note that Full is a greedy algorithm which tries to fetch the highest quality tiles only within the viewport. Hence, there will be more quality variation between subsequent chunks. Thus, the QoE attained by Full is the lowest. 



 \subsubsection{Distribution of the rates}
In Fig.~\ref{fig:rate}, we plot the distribution of the rates $\gamma_k$.  Most number of tiles are fetched at two nearby rates $0.5$ and $0.75$ Mbps. Algorithm 360-ROBUST also fetches a lot of tiles in the highest quality specifically for the later chunks where there is enough bandwidth for chunks to be fetched.

The algorithm BA1 only fetches the tiles within the viewport, and the adjacent tiles at a higher rate. Hence, it may not fetch all the tiles within the set $\mathcal{A}_{\alpha,k}$ with a higher rate. Further, it tries to utilize all the bandwidth to fetch the tiles of the next chunk  at the highest quality rather using the bandwidth to fetch the tiles of the later chunks at a higher rate in an optimal manner. Thus, the performance of BA 1 is worse compared to our proposed algorithm. 


\begin{figure}
\includegraphics[width=60mm, height=40mm]{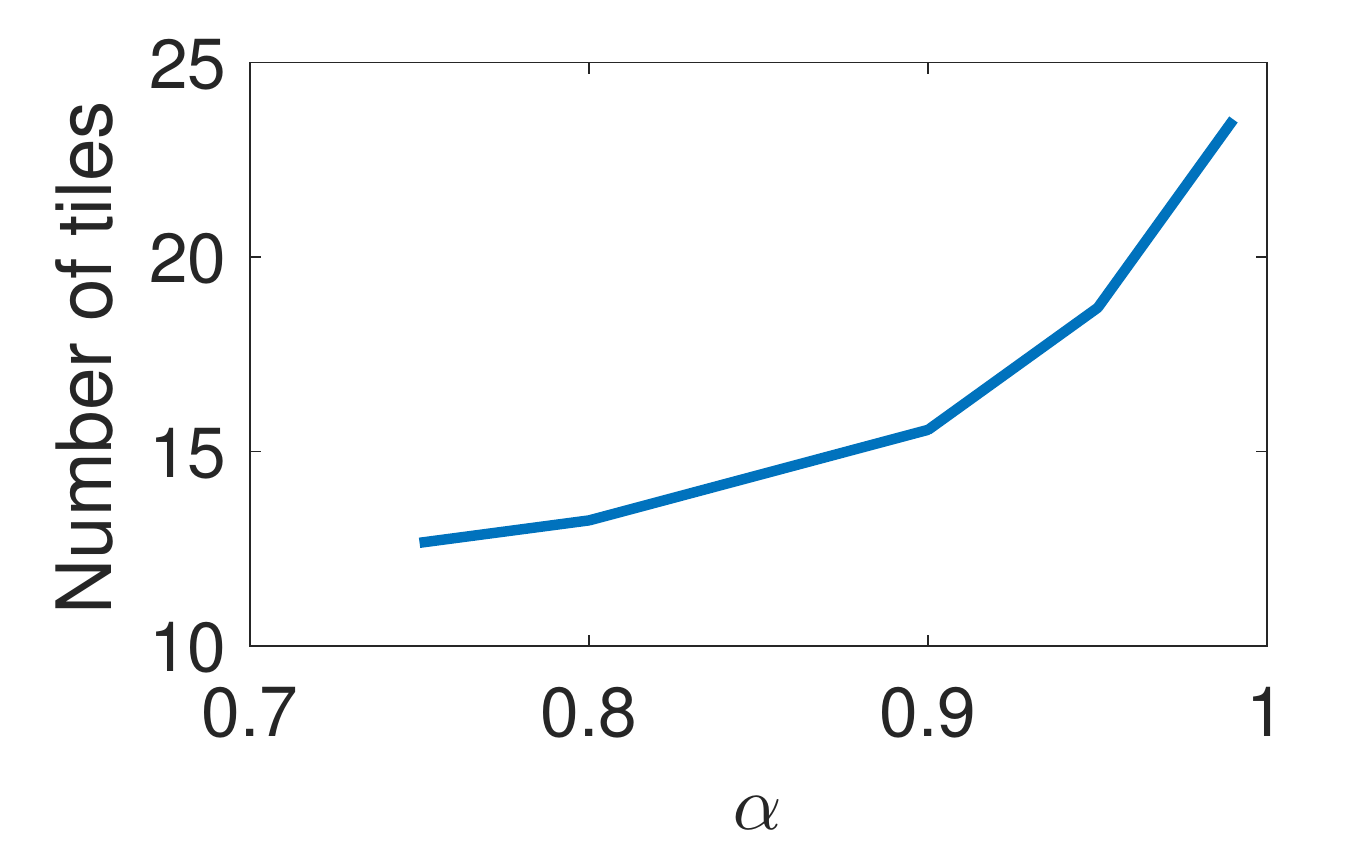}
\vspace{-0.1in}
\caption{The average number of tiles within $|\mathcal{A}_{\alpha,k}|$ in a chunk as a function of $\alpha$. Average is taken over the number of chunks}
\label{fig:alpha}
\vspace{-0.2in}
\end{figure}

\subsubsection{The cardinality of $\mathcal{A}_{\alpha,k}$}
Fig.~\ref{fig:alpha} shows the variation of $|\mathcal{A}_{\alpha,k}|$ with $\alpha$. As $\alpha$ increases, the cardinality increases since in order to provide a rate with a higher probability, more tiles is required to be downloaded. The increment is exponential when $\alpha$ exceeds a threshold. Even when $\alpha=0.95$, less than 60\% of tiles constitute the set $\mathcal{A}_{\alpha,k}$. 

\subsubsection{Effect of the Window Size}
 Fig.~\ref{fig:window} shows that as the window size increases, the QoE increases. However, the QoE decreases after a threshold. Intuitively, if $W$ is large, the algorithm 360-ROBUST may be conservative, and only use excess bandwidth for fetching tiles at a higher quality for the last few chunks. The FoV or bandwidth prediction may also be bad if $W$ is large, specifically, if there is a temporal correlation where the prediction is highly correlated over short time period. On the other hand, if $W$ is too small, the algorithm may become too greedy, leading to a higher quality variation  between the subsequent chunks. $W=5$ performs the best. Note that the baseline algorithms decide over the download rates for the tiles of the immediate next chunk without considering the future chunks. Thus, in the baseline algorithms, $W$ is always $1$. 
 
Fig.~\ref{fig:run_time} shows the average run-times (over different bandwidth traces) with the window size for our proposed algorithm {360-ROBUST}. The run-time is calculated on 8GB-DDR3 RAM and 1.6GHz Intel Core i5 processor, using a C++ code. Fig~\ref{fig:run_time} shows that as the window size increases the run-time increases. This is because the decision space increases with $W$. 
 
 \begin{figure*}
 \vspace{-0.1in}
 \begin{minipage}{0.31\textwidth}
 \includegraphics[trim=0in 0in 0.4in 0in, width=\textwidth]{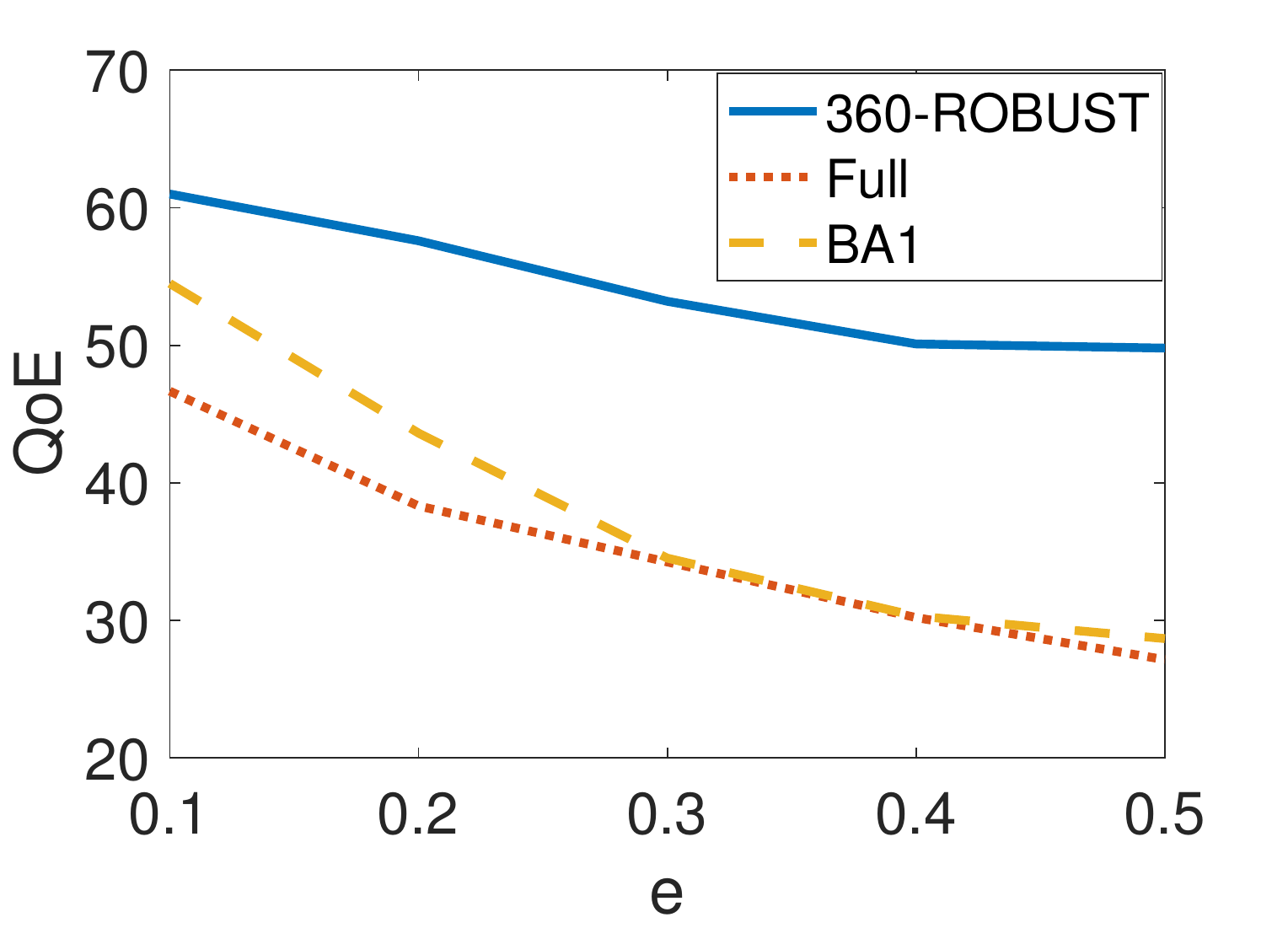}
\caption{The variation of QoE with the bandwidth prediction error ($e$).}
\label{fig:band}
 \end{minipage}
 \begin{minipage}{0.32 \textwidth}
  \includegraphics[trim=0in 0in 0.6in 0in, width=\textwidth]{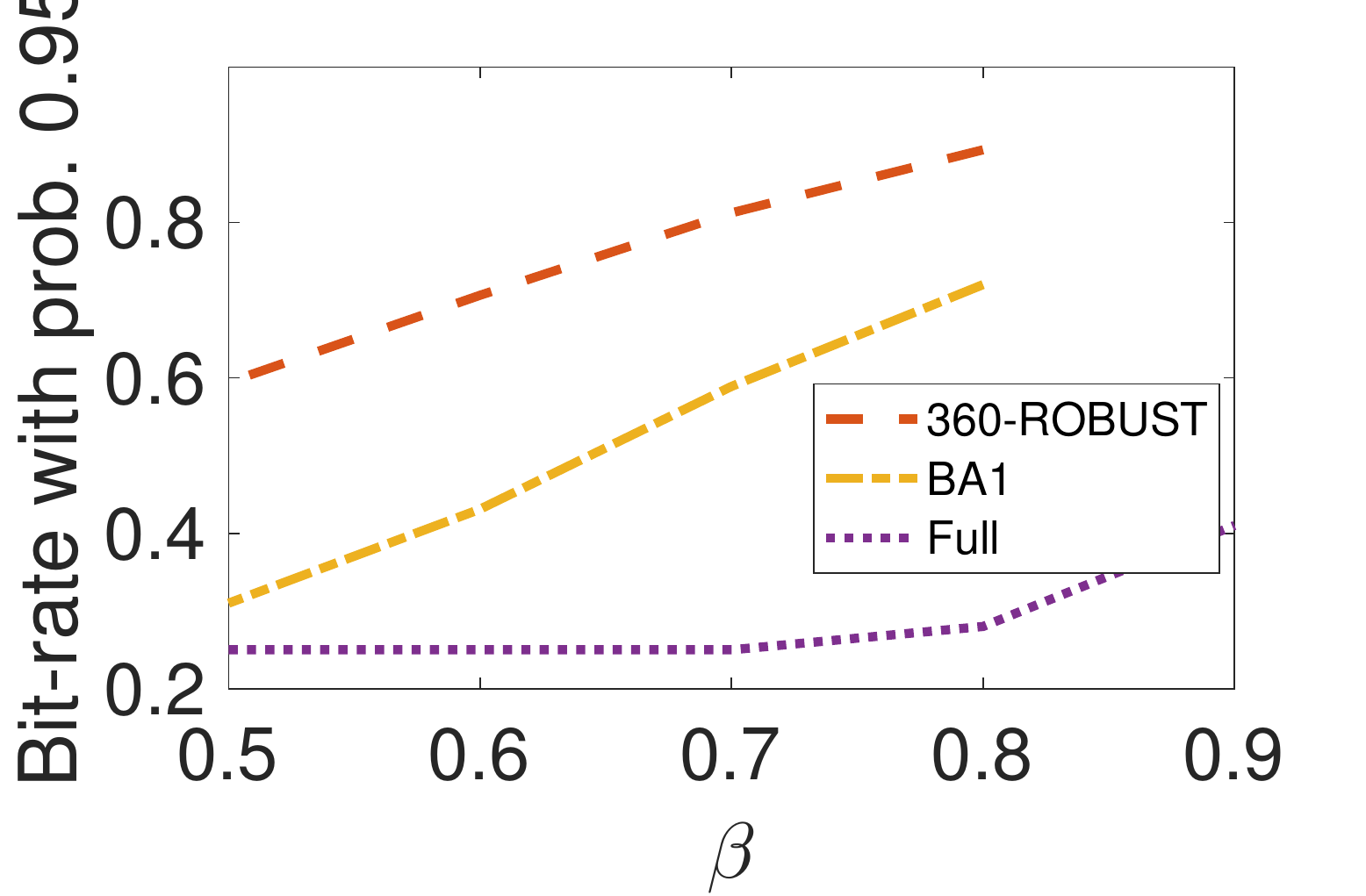}
\caption{The variation of the rate such that the quality degrades below that rate only with probability $0.05$ with FoV prediction error ($\beta$).}
\label{fig:fov}
 \end{minipage}\hfill
 \begin{minipage}{0.32\textwidth}
 \includegraphics[width=\textwidth]{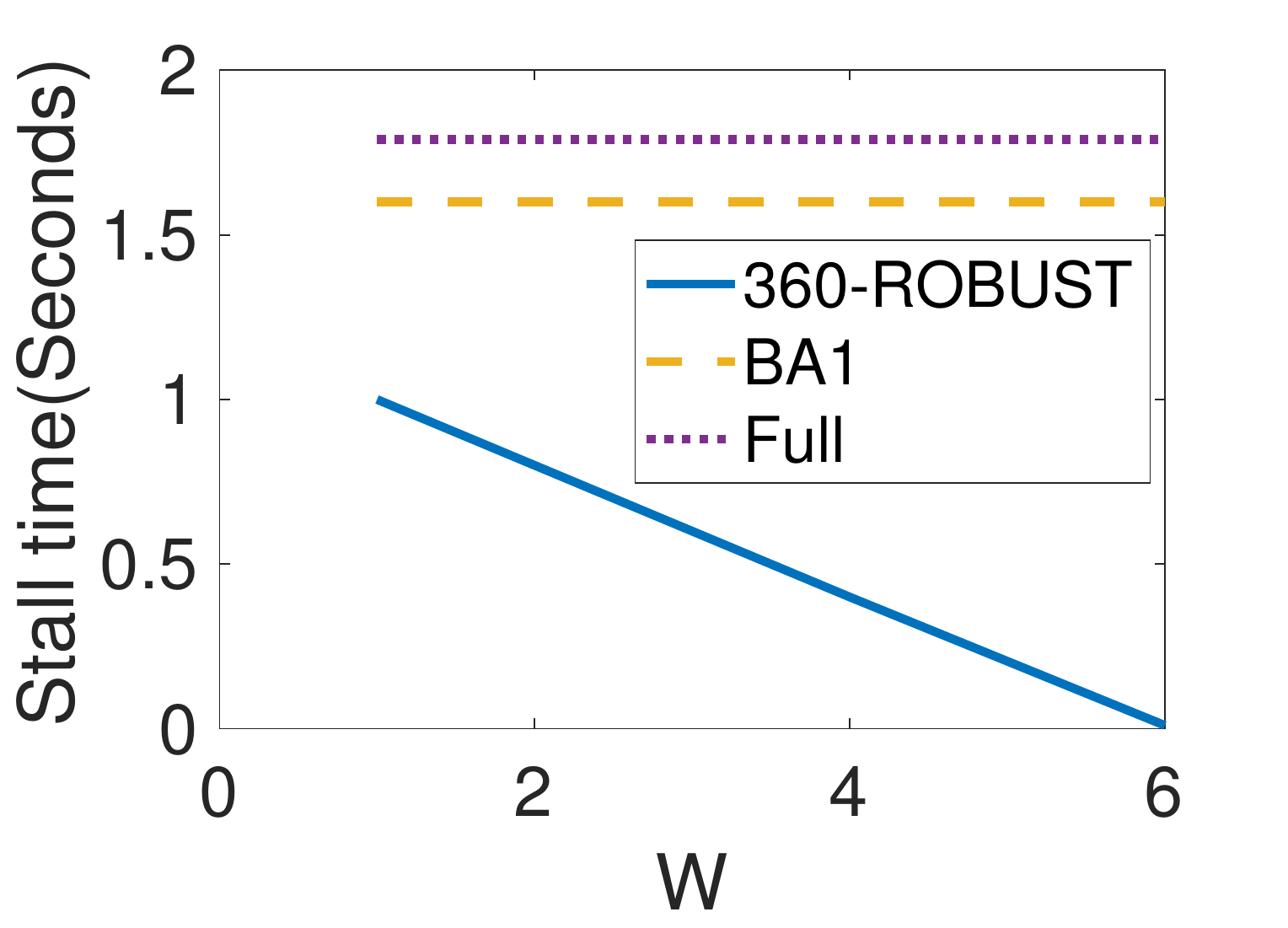}
\vspace{-0.1in}
\caption{Variation of the stall time with window size ($W$).}
\label{fig:stall}
\end{minipage}
 \vspace{-0.2in}
 \end{figure*}
%
 
 \subsubsection{Impact of Bandwidth prediction Error}
In order to study the impact of the bandwidth prediction error we multiply the original bandwidth traces by a factor $(1+p)$ where $p$ is a randomly chosen from an uniform distribution $[-e,e]$. The impact of the variation of QoE with $e$ is shown in Fig.~\ref{fig:band}. As $e$ increases the QoE decreases as the uncertainty increases. However, the decrement is very small compared to the baseline algorithms. This shows the robustness of our proposed algorithm compared to the existing algorithms. Our proposed algorithm 360-ROBUST fetches the tiles  for $W$ chunks ahead. Further, the algorithm reserves the excess bandwidth to download the later chunks. However, the baseline algorithms BA1 and Full utilize the bandwidth fully to download the current chunk. Hence, if there is a large error in estimating the bandwidth, the performance degrades. 
 
\subsubsection{Impact of FoV prediction Error}
In order to show the robustness of the algorithms, we have used the traces of the users' head-position with probability $\beta$ and randomly select a FoV with probability $1-\beta$. We characterize the performance of the algorithms as $\beta$ increases.   Fig.~\ref{fig:fov} shows that the performance of all algorithms degrade as $\beta$ decreases. However, our proposed algorithm 360-ROBUST outperforms all the algorithms. For lower values of $\beta$, our algorithm performs almost $80\%$ better compared to the baselines. This is because the baseline algorithms assume a strong correlation between the head positions at the successive frames. However, our proposed algorithm uses the crowdsourced data. The above shows the robustness of our proposed algorithm. 

The algorithms BA1, and Full do not employ any crowd-source data to compute the FoV distribution, rather they rely only on the head-position prediction. Thus, as the FoV prediction or the head-position prediction error increases, the performance of the baseline algorithm degrade at a faster rate compared to our proposed algorithm. Further, the baseline algorithms fetch the tiles within the current viewport (BA1, Full) and adjacent to it (BA1) for the next chunk. However, they may not fetch all the tiles within the set $\mathcal{A}_{\alpha,k}$ for the chunk $k$. Hence, the baseline algorithms do not provide any guaranteed rate unlike our proposed algorithm. 


%
\subsubsection{Stall Time}
Fig.~\ref{fig:stall} shows that our proposed algorithm again out performs the baseline algorithms. Our analysis reveals that as window size ($W$) increases the stall time decreases.  Intuitively, when $W$ is small the algorithm is greedy, and thus wants to utilize the extra bandwidth to fetch tiles of the current chunk only. Hence, if the future bandwidth is poor it may lead to a stall. Since the baseline algorithms do not use any window, their performance do not depend on the window size. They determine the download rates for the next chunk without considering the future bandwidth, or the future FoV. When $W=5$, our proposed algorithm achieves a stall time almost close to zero. The stall time achieved by the baseline algorithms never reach zero.  



\section{Conclusions and Future Work}\label{sec_concl}
We formulate a robust QoE metric for  tile-based 360 degree video streaming by considering both the FoV and bandwidth variability. 
We consider a QoE metric which maximizes the streaming rate with which a user will watch at least with probability $\alpha$ along with minimizing the stall time, and the variation of rates across the chunk.
%
We formulate the problem as a robust stochastic optimization problem. The problem turns out to be non-convex because of  discrete strategy space. We propose a novel algorithm, {360-ROBUST}, which gives a feasible solution from the optimal solution of a relaxed convex problem. We provide a theoretical optimal gap of the feasible solution provided by {360-ROBUST}. We also propose an online version of the algorithm which adapts to the real-time variation of the FoV and the bandwidth. Our empirical evaluation shows that our proposed algorithms provide 30\% higher QoE compared to the baseline algorithms.

The characterization of the optimality of the online algorithm and the offline algorithm constitutes an interesting future direction.  The characterization of robust algorithm for other QoE metric is also left for the future. 

\bibliographystyle{IEEEtran}
\bibliography{360degreeref}
\end{document}